\documentclass{article}
\usepackage[utf8]{inputenc}

\usepackage{graphicx}
\usepackage{epstopdf}
\usepackage{amsthm}

\newtheorem{df}{Definition}
\newtheorem{lem}[df]{Lemma}
\newtheorem{thm}[df]{Theorem}
\newtheorem{cor}[df]{Corollary}

\title{Hamiltonian cycles and paths  in hypercubes with disjoint faulty edges}

\author{
Janusz Dybizba\'nski, Andrzej Szepietowski \\
\small Institute of Informatics,\\ 
\small Faculty of Mathematics, Physics and Informatics, \\
\small University of Gda\'nsk, 
\small 80-308 Gda\'nsk, Poland \\
\small \texttt{jdybiz@inf.ug.edu.pl, matszp@inf.ug.edu.pl}}

\date{}

\begin{document}

\maketitle

\begin{abstract} 
We consider hypercubes with pairwise disjoint faulty edges. 
An $n$-dimensional hypercube $Q_n$ is an undirected graph with $2^n$ nodes, each labeled with a distinct binary strings of length $n$.
The parity of the vertex is 0 if the number of ones in its labels is even, and is 1 if the number of ones is odd.
Two vertices $a$ and $b$ are connected by the edge iff $a$ and $b$ differ in one position.
If $a$ and $b$ differ in position $i$, then we say that the edge $(a,b)$  goes in direction $i$ and we define the parity of the edge as the parity of the end with 0 on the position $i$. It was already known that $Q_n$ is not Hamiltonian if all edges going in one direction and of the same parity are faulty.

In this paper we show that if $n\ge4$ then all other hypercubes are Hamiltonian. In other words, every cube $Q_n$, with $n\ge4$ and disjoint faulty edges is Hamiltonian if and only if for each direction there are two healthy crossing edges of different parity.
\end{abstract} 


\section{Introduction}\label{intro}
For definitions and notations we follow~\cite{DybSzepCycle, DybSzepPaths}. An $n$-dimensional hypercube (cube), denoted by $Q_n$, is an undirected graph with $2^n$ nodes, each labeled with a distinct binary string $b_{n-1}\dots b_1b_0$, where $b_i\in\{0,1\}$. A vertex $x=b_{n-1}\dots b_i\dots b_0$ and the vertex $x^{(i)}=b_{n-1}\dots \bar{b_i}\dots b_0$ are connected by an edge along dimension $i$, where $\bar{b_i}$ is the negation of $b_i$. The hypercube $Q_n$ is  bipartite, the set of nodes is the union of two sets: nodes of parity 0 (the number of ones in their labels is even), and nodes of parity 1 (the number of ones is odd), and each edge connects nodes of different parity. We shall denote by $Par(x)$ the parity of the vertex $x$.

For any dimension  $i$,  there is a partition of $Q_n$ into two subcubes $Q^i_0=\{x\in Q_n:x_i=0\}$ and $Q^i_1=\{x\in Q_n:x_i=1\}$. We shall often use simpler notations $Q_0$ or $Q_1$, if the dimension is obvious or irrelevant. For any edge $e=(u,v)$, if $u\in Q^i_0$ and $v\in Q^i_1$, then we say that $e$ is crossing and goes in $i$ dimension. We say that $e$ is of parity 0 (or 1) if $u$ is of parity 0 (or 1).

A path or a cycle is \textit{Hamiltonian} if it visits each node in the cube exactly once. The hypercube is Hamiltonian, i.e., it contains a Hamiltonian cycle. Moreover, it is \textit{laceable}, which means that any two vertices of different parity can be connected by a Hamiltonian path. Note that  vertices of the same parity cannot be connected by a Hamiltonian path.  An important property of the hypercube is its fault tolerance. It contains Hamiltonian cycles and paths even if some edges or vertices are faulty.

In this paper we consider Hamiltonian cycles and paths in hypercube $Q_n$ with faulty edges (we shall denote the set of faulty edges  by $F$). Tsai et al.~\cite{Tsai} showed that if $|F|\le n-2$, then $Q_n$ is  Hamiltonian. The hypercube $Q_n$ is $n$-regular, so it is not Hamiltonian if it contains a vertex with incident $n-1$ faulty edges. We shall call such a set of faulty edges a local trap.  Xu et al. \cite{4} proved that if  each vertex is incident with at least two healthy edges and $|F|\le n-1$, then the cube is still Hamiltonian. Chan and Lee~\cite{1} showed  that also with $|F|\le 2n-5$  faulty edges $Q_n$ is Hamiltonian. Shih et al.~\cite{2}  observed that the bound $2n-5$ is optimal, because with $2n-4$ faulty edges it is possible to build another kind of local trap, namely a (healthy) cycle $(u,v,w,x)$, where all edges going out of the cycle from $u$ and $w$ are faulty. This kind of trap is called $f_4$-cycle in~\cite{Li}. Szepietowski~\cite{Sz} observed that these two kinds of traps are examples of a more general scheme, namely subgraphs disconnected halfway.

\begin{df}\label{defdhw}
A proper subgraph $T\subset Q_n$ is disconnected halfway if  half of the nodes of $T$ have parity 0 and either (1) all edges joining the nodes of parity 0 in $T$  with nodes outside $T$, are faulty, or (2) all edges joining the nodes of parity 1 in $T$  with nodes outside $T$, are faulty.
\end{df}

\begin{lem}\label{lem2}(\cite{Sz}) The cube 
$Q_n$ is not Hamiltonian if it has a trap disconnected halfway.
\end{lem}
If $Q_n$ has a node $u$ with $n-1$ faulty incident edges, then $u$ with its neighbor form a subcube of dimension 1 disconnected halfway. We shall denote it by $Q_1$-DHW. If $T$ is a cycle $C_4=(u,v,w,x)$ (subcube of dimension 2), then we obtain the trap described in \cite{2} and mentioned above. Similarly, in $Q_n$, with $n\ge 4$, using  $4(n-3)$ faulty edges we can disconnect halfway  the 3-dimensional subcube $Q_3$. Szepietowski~\cite{Sz} proved that  the cube $Q_n$ with  $n\ge5$ and with $2n-4$ faulty edges   is not Hamiltonian only if it  contains a $Q_1$-DHW or $Q_2$-DHW trap. Liu and Wang \cite{Liu} showed that the same is true for hypercubes with $3n-8$ faulty edges. The bound $3n-8$ is sharp, because with $3n-7$   faulty edges it is possible to build another kind of local trap, namely a (healthy) cycle $(u,v,w,x,y,z)$, where all edges going out of the cycle from $u$, $w$, and $y$ are faulty, see~\cite{Li,Liu,Sz}. Li et al.~\cite{Li} proved that with $n\ge 6$ and with $3n-7$ faulty edges $Q_n$ is not Hamiltonian only if it  contains a $Q_1$-DHW, $Q_2$-DHW, or $C_6$-DHW trap.

In~\cite{Sz}, another type of trap was described which is not a set disconnected halfway. Consider a node $u$ that  has three neighbors $v$, $w$ and $x$, each of degree 2, see Fig.~\ref{figCLtrap}. Then, the cube $Q_n$ is not Hamiltonian.
\begin{figure}[htb]
\centering
\includegraphics{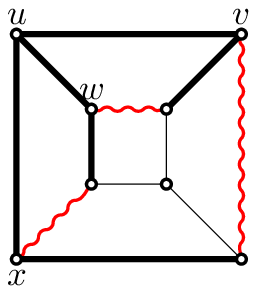}
\caption{A claw trap.}
\label{figCLtrap}
\end{figure}
We call them  claw traps. In $Q_3$ claw consists of three disjoint faulty edges. If $n\ge 4$, the faulty edges in Claw are not disjoint. 
 Each of the trap mentioned above consists of small collection of faulty edges where the number of faulty edges is bounded by a linear function.

In this paper we consider cubes $Q_n$ with pairwise disjoint faulty edges, i.e. for any two faulty edges $e,f\in F$, $e\cap f=\emptyset$. In~\cite{DybSzepCycle} authors  described all nonhamiltonian cubes with $n=3$ and $n=4$. From their results it follows that if $F$ consists of  pairwise disjoint edges, then the cube $Q_3$ is not Hamiltonian if and only if it contains a $Q_2$-DHW or a Claw trap, and that $Q_4$ is not Hamiltonian if and only if it contains a $Q_3$-DHW trap. It is easy to see that  $Q_3$ is laceable if and only if $|F| \le 1$. Using computer we checked that $Q_4$ with disjoint faulty edges is not laceable only if it contains either $Q_3$-DHW or six parallel faulty edges.

If the set $F$ consists of disjoint edges, then $|F|\le 2^{n-1}$ and only one trap of the kind DHW is possible. Namely $Q_{n-1}$-DHW. 

\begin{df}
We say that $Q_n$ with $n\ge 3$, has a SCDHW (subcube disconnected halfway) along dimension $i$,  if all crossing edges of parity  0 (or 1) going in dimension $i$ are faulty.
\end{df}
Suppose that $F$ consists of disjoint edges and $Q_n$ with $n\ge 3$, has a SCDHW along dimension $i$. We may assume that all vertices of parity 0 in $Q^i_0$ are cut off from the vertices in $Q^i_1$. Then  there are faulty edges neither in $Q^i_0$ nor in $Q^i_1$. However some other edges going in dimension $i$ can be faulty. By Lemma~\ref{lem2}, $Q_n$ is not Hamiltonian, hence, it is not laceable. Moreover, if $Q_n$ contains a Hamiltonian path, then one of its ends is of parity 0 and belongs to $Q^i_0$ and the other end is of parity 1 and belongs to $Q^i_1$, see 
Lemma~10 in~\cite{DybSzepPaths}.
 
\begin{df}
We say that $Q_n$ has DTBCE (dimension with two balanced crossing edges), if there is a dimension $i$ and two vertices $u, w\in Q^i_0$ of different parity  such that: all crossing edges in dimension $i$, except two $(u,u^{(i)})$ and $(w,w^{(i)})$, are faulty.
\end{df}
Note that it is possible that the vertices $u$ and $w$ (or $u^{(i)}$ and $w^{(i)}$) form a faulty edge. Note also that if $F$ consists of disjoint edges then no other faulty edges are possible, hence, there is at most one faulty edge in $Q^i_0$ and at most one  in $Q^i_1$. The following lemma is easy to prove,
\begin{lem}\label{DTBCE1} The cube  $Q_n$ with DTBCE:

 (i)  has a Hamiltonian cycle, even if $(u,w)\in F$ and $(u^{(i)},w^{(i)})\in F$
 
 (ii)  is not laceable. There is neither a HP going from $u$ to $w$ nor a HP going from $u^{(i)}$ to $w^{(i)}$.
\end{lem}

In this paper we show that  if $n\ge 4$ and $F$ consists of pairwise disjoint edges, then only the cubes which has SCDHW are not Hamiltonian. In other words $Q_n$ with $n\ge 4$, is Hamiltonian if and only if for each dimension $i$, there is a healthy crossing edge of parity 0 and a healthy crossing edge of parity 1. Moreover we show that $Q_n$ is not laceable if and only if it contains either SCDHW or DTBCE. In other words $Q_n$  is laceable if and only if for each dimension $i$, there are at least three healthy edges not of the same parity.

\section{Preliminaries}
Let $f_0$ denote the number of faulty edges in $Q_0$ and $f_1$ denote the set of faulty edges in $Q_1$. By $dist(A,B)$, we denote the distance between the vertices $A$ and $B$.
We shall use abbreviations:  
\begin{itemize}
\item[-] HP $x\to y$ to denote a Hamiltonian path from $x$ to $y$, 
\item[-] HP $x\to_F y$ to denote a Hamiltonian path from $x$ to $y$ passing through one faulty edge,
\item[-] $x\not\to y$ to denote that there are no HP going from $x$ to $y$,
\item[-] $x\not\to_F y$ to denote that there are no HP going from $x$ to $y$ and passing through one faulty edge.
\end{itemize}



\begin{lem}\label{L2} (\cite{Tsai})
If $n\ge 3$ and $|F|\le n-2$, then $Q_n$ is laceable. 
\end{lem}



\begin{lem}\label{L1} (\cite{LW}, see also~\cite{Tsai}) 
Suppose that  $x$, $y$,  and  $v$ are three vertices in  $Q_n$ with $n\ge 2$, and that $Par(x)=Par(y)\ne Par(v)$. Then there is a HP $x\to y$ in $Q_n-v$.
\end{lem}

\begin{lem}\label{L3} (\cite{Dvorak}) Suppose that $p$, $q$, $r$, and $s$ are four different vertices in $Q_n$ and  $Par(p)=Par(q)=0$, $Par(r)=Par(s)=1$. Then there is a partition of $Q_n$ into two disjoint paths $p\to r$ and $q\to s$.
\end{lem}

\begin{lem}\label{L4} 
For any two  vertices $x$, $y$,  with  $Par(x)\ne Par(y)$ and an edge $e\ne (x,y)$, there is HP $x\to y$ passing through $e$.
\end{lem}
\begin{proof}
If neither  $x$ nor $y$ is incident to $e$, then the lemma follows from Lemma~\ref{L3}. If $x$ or $y$ is incident to $e$, then the lemma follows from Lemma~\ref{L1}. 
\end{proof}

We shall use the following notations. The vertex $x$ is free if it is not incident to any faulty edges. For two dimensions  $i\ne j$,  by $Q^{i,j}_{00}$ we denote the set $\{x\in Q_n:x_i=0\;{\rm and}\; x_j=0\}$. Similarly we define $Q^{i,j}_{01}$, $Q^{i,j}_{10}$, and $Q^{i,j}_{11}$.

\section{Partition}

\begin{lem}\label{Partition}
In $Q_n$, with $n\ge 5$ and disjoint faulty edges, there is a dimension $m$ such that both $Q^m_0$ and $Q^m_1$ have neither SCDHW nor DTBCE.
\end{lem}

In order to prove the lemma we shall need two lemmas. 

\begin{lem}\label{Partition0}
Suppose that in $Q_n$ with $n\ge 4$, there are two different dimensions, $i$ and $k$.

(i) If there are two vertices $u,v\in Q^{i,k}_{00}$ or   $u,v\in Q^{i,k}_{01}$ such that: $Par(u)\ne Par (v)$ and the edges $(u,u^{(k)})$ and $(v,v^{(k)})$ are healthy. Then there is no SCDHW along dimension $k$ in  $Q^i_0$.

(ii) If there are two vertices $u,v\in Q^{i,k}_{10}$ or $u,v\in Q^{i,k}_{11}$ such that: $Par(u)\ne Par (v)$ and the edges $(u,u^{(k)})$ and $(v,v^{(k)})$ are healthy. Then there is no SCDHW along dimension $k$ in  $Q^i_1$.
\end{lem}
\begin{proof}
(i) If the subcube $Q^i_0$ has SCDHW along dimension $k$, then both in $Q^{i,k}_{00}$ and in  $Q^{i,k}_{01}$ either all vertices of parity 0 or all vertices of parity 1 are incident to faulty edges going in dimension $k$.
\end{proof}

\begin{lem}\label{Partition1}
Suppose that $n\ge 4$ and the subcube $Q^i_0$ (or $Q^i_1$) has SCDHW along dimension $k$, with $k\ne i$. We may assume that $i=1$ and $k=2$. Then:

(i) for each $j>2$, neither $Q^j_0$ nor $Q^j_1$ has SCDHW along dimension $m=1$.

(ii) for each $j>2$, neither $Q^j_0$ nor $Q^j_1$ has SCDHW along dimension $m$, with $m>2$ and $m\ne j$.

(iii) neither $Q^2_0$ nor $Q^2_1$ has SCDHW along dimension $m$, with $m>2$. 
\end{lem}
\begin{proof}
We may assume that  every vertex of parity 0 in $Q^{1,2}_{00}$ is joined by a faulty edge with its neighbour in $Q^{1,2}_{01}$ (and every vertex of parity 1 in $Q^{1,2}_{01}$ is joined by a faulty edge with its neighbour in $Q^{1,2}_{00}$). Consider  partition of $Q_n$ into $Q^j_0$ and $Q^j_1$, for any dimension $j\ge3$, let say $j=3$, see Fig.~\ref{part}.
\begin{itemize}
\item[(i)] First consider two vertices $u=000^{n-2}$ and $v=010^{n-2}$ in $Q^{1,3}_{00}$. They are of different parity and are incident to faulty edges going in dimension 2, so the edges $(u,u^{(1)})$ and $(v,v^{(1)})$ are healthy.  Hence, by Lemma~\ref{Partition0}, $Q^3_0$ has no SCDHW along dimension $m=1$. Next consider two vertices $x=00110^{n-4}$ and $y=01110^{n-4}$ in $Q^{1,3}_{01}$. They are of different parity and  the edges $(x,x^{(1)})$ and $(y,y^{(1)})$ are healthy.  Hence, by Lemma~\ref{Partition0}, $Q^3_1$ has no SCDHW along dimension $m=1$.
\item[(ii)] First suppose that $Q^3_0$ has SCDHW along dimension $m>3$, let say $m=4$. There are two vertices $u=00000^{n-4}$ and $v=01000^{n-4}$ in $Q^{3,4}_{00}$. They are of different parity and the edges $(u,u^{(4)})$ and $(v,v^{(4)})$ are healthy.  Hence, by Lemma~\ref{Partition0}, $Q^3_0$ has no SCDHW along dimension $m=4$. In order to prove that $Q^3_1$ does not have SCDHW along dimension 4, consider two vertices $u=00110^{n-4}$ and $v=01110^{n-4}$ in $Q^{3,4}_{11}$. They are of different parity and the edges $(u,u^{(4)})$ and $(v,v^{(4)})$ are healthy. 
\item[(iii)] Suppose that $Q^2_0$ has SCDHW along dimension $m>2$, let say $m=3$. Then either each vertex of the form 000... and of parity 0 is incident to faulty edge in dimension 3, or each vertex of the form 001...  and of parity 0 is incident to faulty edges going in dimension 3. This is not possible because all vertices of the form 00... and of parity 0 are incident to faulty edges going in dimension 2. Similarly, we can show that in $Q^2_1$ there are no SCDHW along dimension $m$.
\end{itemize} 
\end{proof}

\begin{figure}[htb]
\centering
\includegraphics{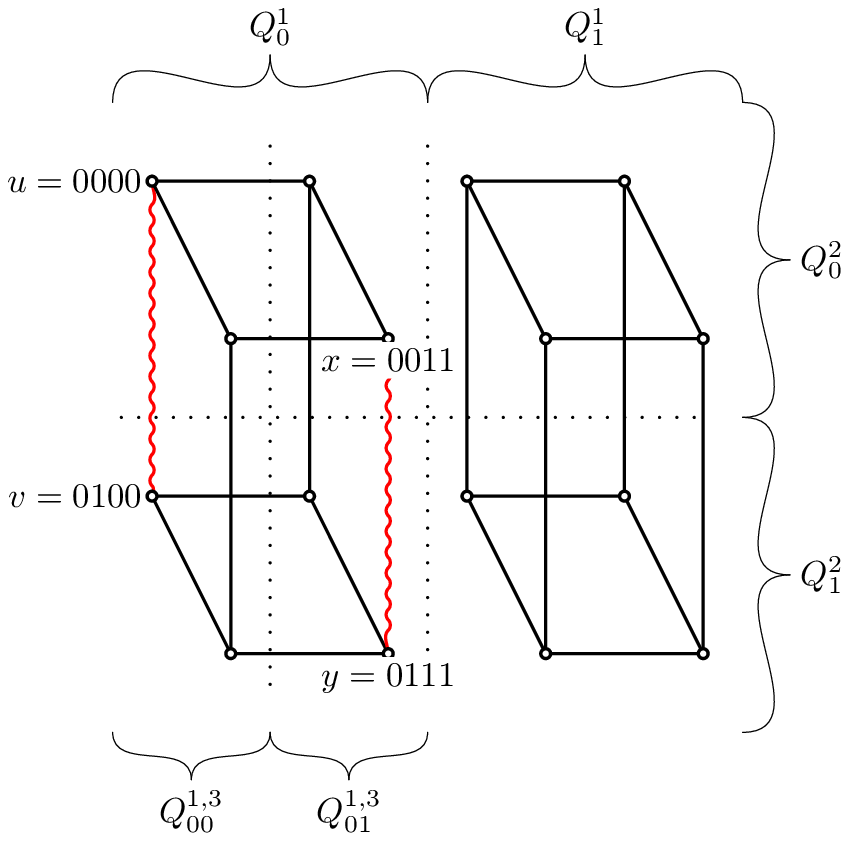}
\caption{$Q_4$ with SCDHW along dimension 2 in $Q^1_0$.}
\label{part}
\end{figure}

\begin{proof}[Proof of Lemma~\ref{Partition}.] 
First suppose that for some dimension $i$, the subcube  $Q^i_0$ or $Q^i_1$ has DTBCE along some dimension  $j\ne i$. We may assume that the subcube $Q^1_0$ has DTBCE along dimension 2. Consider  partition of $Q_n$ into $Q^2_0$ and $Q^2_1$. Then neither $Q^2_0$ nor $Q^2_1$ has  SCDHW or DTBCE. In order to prove this, observe that in the subcube $Q^{1,2}_{00}$ all vertices of parity 0 except one (there are $2^{n-3}-1\ge 3$ such vertices) and all vertices of parity 1 except one are incident to faulty edges in dimension 2, hence they are free inside $Q^2_0$. Similarly, in the subcube $Q^{1,2}_{01}$ at least three vertices of parity 0 and at least three vertices of parity 1 are free inside $Q^2_1$. Observe that in $Q^2_0$:
\begin{itemize}
\item DTBCE  is not possible, because a subcube with DTBCE has at most two free vertices of parity 0.
\item SCDHW along dimension 1 is not possible because in such a case  all vertices of parity 0 (or 1) in $Q^{12}_{00}$ (there are $2^{n-3}\ge 4$ such vertices) are incident to faulty edges inside $Q^2_0$.
\item SCDHW along dimension $k>2$, let say $k=3$,  is not possible, because in such a case  all vertices of the form $000...$ and parity 0 (or 1)  (there are $2^{n-4}\ge 2$ such vertices) are incident to faulty edges inside $Q^2_0$.
\end{itemize}
Similarly we can show that $Q^2_1$ has neither DTBCE nor SCDHW.

Next suppose that for each dimension $i$, neither $Q^i_0$ nor $Q^i_1$ has DTBCE, and that for some $i$, $Q^i_0$ has SCDHW along dimension $k$ with $k\ne i$. 
We may assume that $Q^1_0$ has SCDHW along direction 2.
Now, consider the partition along direction 2. By Lemma~\ref{Partition1}, the subcubes $Q^2_0$ or $Q^2_1$ may have SCDHW only along direction 1. If they have, then consider partition along direction 3. By Lemma~\ref{Partition1}, neither $Q^3_0$ nor $Q^3_1$ has any SCDHW.
\end{proof}

\section{Main results}

\begin{thm}\label{MT1}
The cube $Q_n$ with $n\ge 4$ and pairwise disjoint faulty edges, is not Hamiltonian if and only if it has SCDHW.
\end{thm}

\begin{thm}\label{MT2}
The $Q_n$ with $n\ge 4$ and disjoint faulty edges is laceable if for every dimension $i$, there are at least three healthy crossing edges, and they are not of the same parity.
\end{thm}

These two theorems will follow from the lemma:

\begin{lem}\label{Thm1} 
Suppose that in $Q_n$ with  $n\ge 4$, the set of faulty edges $F$ consists of  pairwise disjoint edges and there is neither SCDHW nor DTBCE. Then:

(La1) For any two vertices $A,B\in Q_n$ with $Par(A)\ne Par(B)$, there is a HP $A\to B$,

(La2) For any two vertices $A,B\in Q_n$ with $Par(A)\ne Par(B)$, if $|F|\ge 2$ or $|F|=1$ and $A$, $B$ do not form the edge $(A,B)\in F$, then there is a faulty edge $e$ and a HP $A\to B$ passing through $e$ and not through any other faulty edges.
\end{lem}
\begin{proof}
If $|F|\le 1$, then the theorem follows from Lemmas~\ref{L4} and~\ref{L2}. So we assume that $|F|\ge 2$. Moreover, without loss of generality we may assume that if $A$ and $B$ form an edge, then this edge is healthy. Indeed, any HP $A\to B$ (or $A\to_F B$) does not use the edge $(A,B)$. We may also assume that  $Par(A)=0$ and $Par(B)=1$.

We prove the theorem by induction on $n$. The case $n=4$ we have checked by computer. We have run the program which checks the all 4-dimensional hypercubes with disjoint faulty edges.

Suppose that $n\ge 5$ and that the theorem is valid for smaller cubes. By Lemma~\ref{Partition}, there is a dimension $m$, let say $m=1$, such that after partition of $Q_n$ into $Q^1_0$ and $Q^1_1$, neither $Q^1_0$ nor $Q^1_1$ has SCDHW or DTBCE. Hence, by the induction hypothesis, both $Q^1_0$ and $Q^1_1$  satisfy (La1) and (La2).
\bigbreak
\noindent (La1).

\noindent 
Case 1. $A\in Q_0$ and $B\in Q_1$. Because $Q_n$ has no SCDHW, there is a healthy crossing edge $(u,v)$ of parity 1. By the induction hypothesis, there is a HP $A\to u$ in $Q_0$ and a HP $v\to B$ in $Q_1$. We can build a HP $A\to B$ in $Q_n$ using these two paths and the edge $(u,v)$, see Fig.~\ref{icase1}.
\begin{figure}[ht]
\begin{minipage}[b]{0.45\linewidth}
    \includegraphics[width=\textwidth]{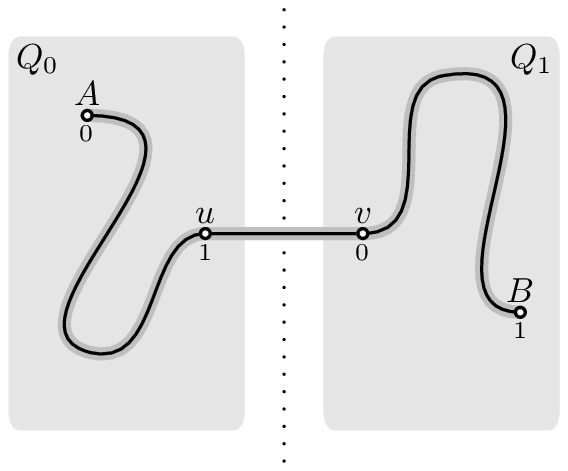}
    \caption{Case 1.}
    \label{icase1}
\end{minipage}
\hfill
\begin{minipage}[b]{0.45\linewidth}
    \includegraphics[width=\textwidth]{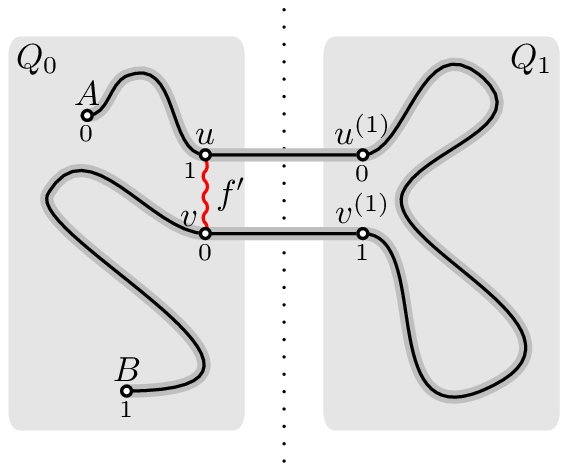}
    \caption{Subcase 2.1}
    \label{icase21}
\end{minipage}
\end{figure}

\smallbreak
\noindent
Case 2. $A,B\in Q_0$.

Subcase 2.1. There is a faulty  edge $f\in F$ with $f\in Q_0$. By the induction hypothesis, there is HP $A\to_F B$ in $Q_0$  passing through a faulty edge $f'=(u,v)$ (it is possible that $f'\ne f$), and there is HP $u^{(1)} \to v^{(1)}$ in $Q_1$. The crossing edges $(u,u^{(1)})$, $(v,v^{(1)})$ are healthy,  because they are incident to the faulty edge $(u,v)$, see Fig.~\ref{icase21}.

Subcase 2.2. There are no faulty edges in $Q_0$. Because $Q_n$ has neither SCDHW nor DTBCE, there are two healthy crossing edges  $(x,x^{(1)})$, $(y,y^{(1)})$ such that: $x,y\in  Q_0$, $Par(x)\ne Par(y)$, and $\{x,y\}\ne \{A,B\}$. We may assume that $Par(A)\ne Par(x)$. By the induction hypothesis, there is HP $x^{(1)} \to y^{(1)}$ in $Q_1$. If $\{A,B\}\cap\{x,y\}=\emptyset$, then by Lemma~\ref{L3}, there is partition of $Q_0$ into two paths $A\to x$ and $B\to y$, see Fig.~\ref{icase22a}.

\begin{figure}[ht]
\begin{minipage}[b]{0.45\linewidth}
    \includegraphics[width=\textwidth]{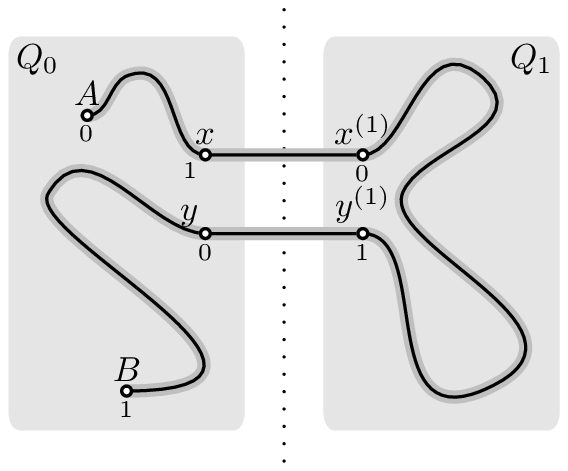}
    \caption{Subcase 2.2, $\{A,B\}\cap\{x,y\}=\emptyset$}
    \label{icase22a}
\end{minipage}
\hfill
\begin{minipage}[b]{0.45\linewidth}
    \includegraphics[width=\textwidth]{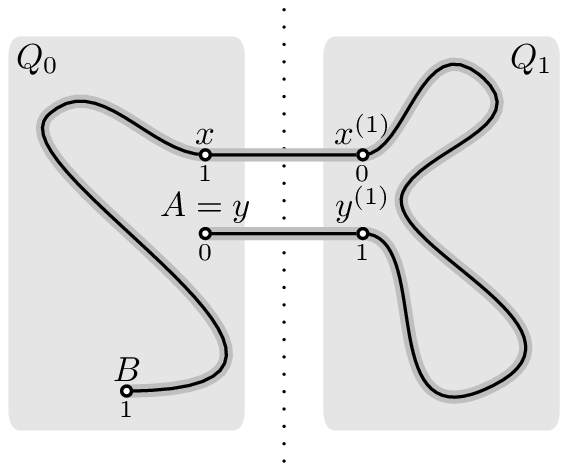}
    \caption{Subcase 2.2, $\{A,B\}\cap\{x,y\}\ne\emptyset$}
    \label{icase22b}
\end{minipage}
\end{figure}
If $\{A,B\}\cap\{x,y\}\ne\emptyset$ then we may assume that $A=y$ and $B\ne x$, and by Lemma~\ref{L1}, there is a HP $x\to B$ in $Q_0-A$, see Fig.~\ref{icase22b}.

\bigbreak
\noindent (La2).

\noindent Case 3. $A\in Q_0$ and $B\in Q_1$.

Subcase 3.1. There is a faulty crossing edge $(u,v)$ of parity 1. The Hamiltonian path in $Q_n$ goes from $A$ to $u$, next crosses to $v$, and then goes from $v$ to $B$, see Fig.~\ref{iicase11}.
\begin{figure}[ht]
\begin{minipage}[b]{0.45\linewidth}
    \includegraphics[width=\textwidth]{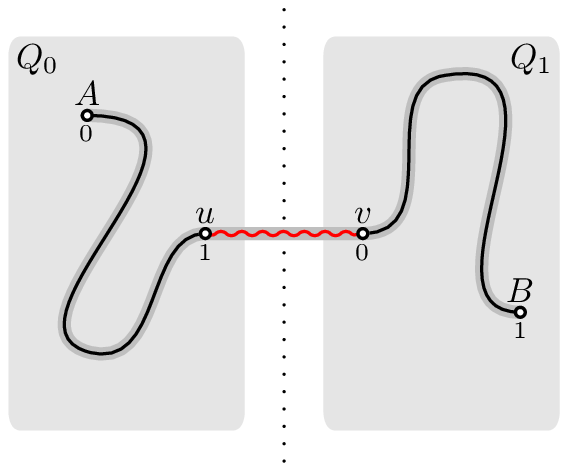}
    \caption{Subcase 3.1}
    \label{iicase11}
\end{minipage}
\hfill
\begin{minipage}[b]{0.45\linewidth}
    \includegraphics[width=\textwidth]{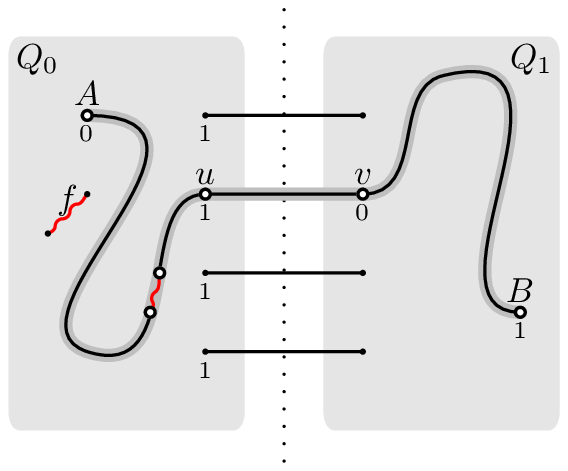}
    \caption{Subcase 3.2}
    \label{iicase12}
\end{minipage}
\end{figure}

Subcase 3.2. All crossing edges of parity 1 are healthy and there is a faulty edge $f=(y,z)$ in  $Q_0$. In this case we choose a healthy crossing edge $(u,v)$ with $u\in Q_0$, $Par(u)=1$ and $u\notin f$. By the induction hypothesis, there is a HP $A\to_F u$ in $Q_0$,
and  there is a HP $v\to B$ in $Q_1$, see Fig.~\ref{iicase12}.

Subcase 3.3. All crossing edges of parity 1 are healthy and   there are faulty edges neither in $Q_0$ nor in $Q_1$. Because $|F|\ge 2$, there are at least two  faulty crossing edges and we can choose an edge $(x,x^{(i)})$ such that $Par(x)=0$ and $x\ne A$. Take two crossing edges $(y,y^{(i)})$, and $(z,z^{(i)})$ with $y,z\in Q_0$ and $Par(y)=Par(z)=1$. There is a partition of $Q_0$ into two disjoint paths $A\to y$ and $x\to z$. If $B\ne x^{(i)}$, then by Lemma~\ref{L3}, we use  a partition of $Q_1$ into two disjoint paths $z^{(i)}\to B$ and $y^{(i)}\to x^{(i)}$, see Fig.~\ref{iicase13a}. If $B=x^{(i)}$, then by Lemma~\ref{L1}, we use a HP $y^{(i)}\to z^{(i)}$ in $Q_1-B$, see Fig.~\ref{iicase13b}.
\begin{figure}[ht]
\begin{minipage}[b]{0.45\linewidth}
    \includegraphics[width=\textwidth]{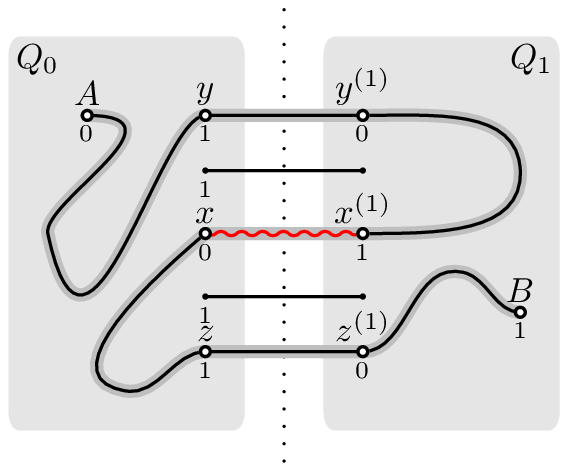}
    \caption{Subcase 3.3. $B\ne x^{(i)}$}
    \label{iicase13a}
\end{minipage}
\hfill
\begin{minipage}[b]{0.45\linewidth}
    \includegraphics[width=\textwidth]{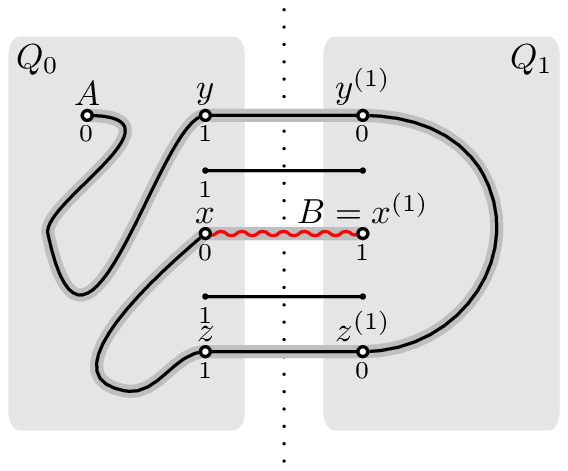}
    \caption{Subcase 3.3. $B=x^{(i)}$}
    \label{iicase13b}
\end{minipage}
\end{figure}

\smallbreak
\noindent Case 4. $A,B\in Q_0$

Subcase 4.1. There are faulty crossing edges. By the induction hypothesis, there is a HP $A\to B$, in $Q_0$. There are two neighbours $y$ and $z$ on the path such that the crossing edge $(y,y^{(1)})\in F$ and the edge $(z,z^{(1)})\notin F$. There is a HP $y^{(i)}\to z^{(i)}$ in $Q_1$, see Fig.~\ref{iicase21}.
\begin{figure}[ht]
\begin{minipage}[b]{0.45\linewidth}
    \includegraphics[width=\textwidth]{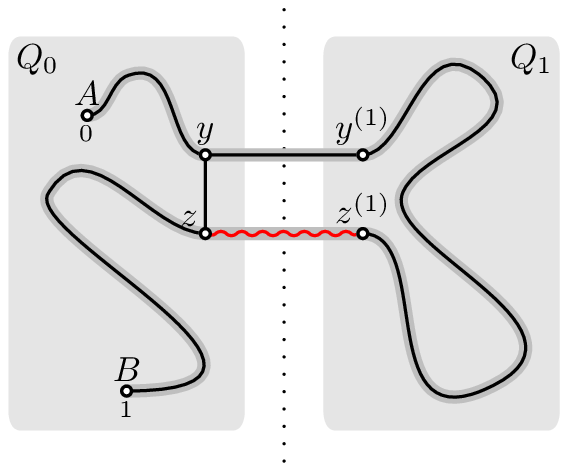}
    \caption{Subcase 4.1}
    \label{iicase21}
\end{minipage}
\hfill
\begin{minipage}[b]{0.45\linewidth}
    \includegraphics[width=\textwidth]{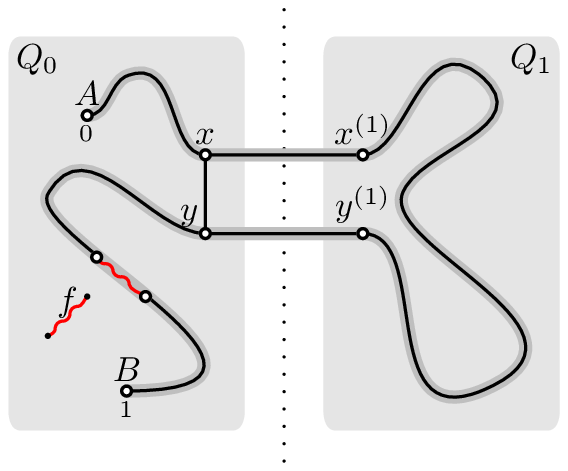}
    \caption{Subcase 4.2}
    \label{iicase22}
\end{minipage}
\end{figure}

Subcase 4.2. There are no faulty crossing edges and there is a faulty edge $f\in Q_0$ with $f\ne (A,B)$. By the induction hypothesis, there is a HP $A\to_F B$ in $Q_0$. 
Choose a healthy edge $(x,y)$ on the path. There is a HP $x^{(i)}\to y^{(i)}$ in $Q_1$, see Fig.~\ref{iicase22}.

Subcase 4.3. There are  no faulty crossing edges and no faulty edges in $Q_0$. Then there are at least two faulty edges in $Q_1$. By the induction hypotheses, there is a HP $A\to B$ in $Q_0$. Choose a (healthy) edge $(x,y)$ on the path. There is a HP $x^{(i)}\to_F y^{(i)}$ 
in $Q_1$, see Fig.~\ref{iicase23}.

\begin{figure}
    \centering
    \includegraphics{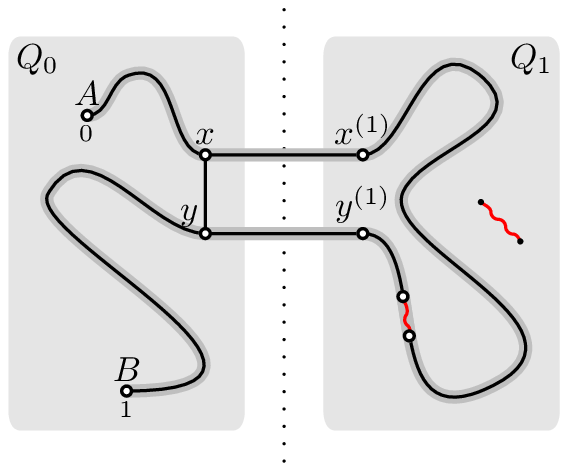}
    \caption{Subcase 4.3}
    \label{iicase23}
\end{figure}
\end{proof}

\begin{proof}[Proof of Theorem~\ref{MT1}.]
By Lemma~\ref{lem2}, $Q_n$ is not Hamiltonian if it has SCDHW. Suppose now that  $Q_n$ is not Hamiltonian,  so it is not laceable and by Lemma~\ref{Thm1}, it  has SCDHW or  DTBCE. However,  by Lemma~\ref{DTBCE1}, $Q_n$ is Hamiltonian if it has DTBCE. 
\end{proof}

\begin{cor} If $n\ge 4$, faulty edges $F$ are pairwise  disjoint, and $|F|<2^{n-2}$, then $Q_n$ is Hamiltonian and laceable.
\end{cor}

\begin{cor} If $n\ge 4$, faulty edges $F$ are pairwise  disjoint, and there are two nonparallel faulty edges, then $Q_n$ is Hamiltonian.
\end{cor}

\begin{cor} If $n\ge 4$, faulty edges $F$ are pairwise  disjoint, and there are three faulty edges of three different dimensions, then $Q_n$ laceable.
\end{cor}


\end{document}